\documentclass[pra,longbibliography,amsmath,amssymb,11pt,nofootinbib]{revtex4-2}
\usepackage{graphicx}
\usepackage{graphicx}
\usepackage{dcolumn}
\usepackage{bm}
\usepackage{graphics}
\usepackage{slashed}
\usepackage{amssymb}
\usepackage{natbib}
\usepackage{amsmath}
\usepackage{amsthm}
\usepackage{url}
\usepackage[usenames,dvipsnames,svgnames,table]{xcolor}
\usepackage[utf8]{inputenc}
\usepackage{boldline}
\usepackage[T1]{fontenc}
\usepackage[mathscr]{eucal}
\usepackage{array, amscd}
\usepackage{multirow}
\usepackage{diagbox}
\usepackage{makecell}
\usepackage[all]{xy}

\makeatletter
\newcommand{\tpitchfork}{%
  \vbox{
    \baselineskip\z@skip
    \lineskip-.52ex
    \lineskiplimit\maxdimen
    \m@th
    \ialign{##\crcr\hidewidth\smash{$-$}\hidewidth\crcr$\pitchfork$\crcr}
  }%
}
\makeatother

\newcommand{\bea}{\begin{eqnarray}}
\newcommand{\ena}{\end{eqnarray}}
\newcommand{\bean}{\begin{eqnarray*}}
\newcommand{\enan}{\end{eqnarray*}}

\newcommand{\KN}{\mathbin{\bigcirc\mspace{-15mu}\wedge\mspace{3mu}}}
 
\newtheorem{prop}{Proposition}[section]

\newtheorem{cor}{Corollary}[prop]

\begin{document}

 \title{Remarks on the algebraic structure of $(2,2)$ double forms}
 
 \author{E. Goulart}
 \email{egoulart@ufsj.edu.br}
\affiliation{Federal University of S\~ao Jo\~ao d'El Rei, C.A.P. Rod.: MG 443, KM 7, CEP-36420-000, Ouro Branco, MG, Brazil
}

\author{J. E. Ottoni}
 \email{jeottoni@ufsj.edu.br}
\affiliation{Federal University of S\~ao Jo\~ao d'El Rei, C.A.P. Rod.: MG 443, KM 7, CEP-36420-000, Ouro Branco, MG, Brazil
}

\date{\today}

\begin{abstract}
We study the algebraic features of covariant tensors of valence four containing two blocks of skew indices. After a rather general treatment, we specialize ourselves to four-dimensional spacetimes and discuss several complementary aspects of these objects. In particular, we focus our attention on the corresponding invariant subspaces and generalise previous relations such as the Ruse-Lanczos identity, the Bel-Matte decomposition and the Lovelock-like quadratic identities. We conclude pointing out some possible applications of the formalism.  
\end{abstract} 
 \maketitle

\section{Introduction}

Covariant tensors of valence four containing two blocks of skew indices are called $(2,2)$ double forms. The latter pop-up everywhere in spacetime physics, come up with a myriad of algebraic facets and have a variety of physical applications. Well known examples in a four-dimensional spacetime  include the Riemann curvature tensor, the Weyl conformal tensor and the Levi-Civita totally anti-symmetric tensor \cite{gravitation, Hawking,Wald,Stephani}. Other geometric-inspired realizations include Kulkarni-Nomizu products of rank two tensors \cite{Besse, Gallot}, the traceless Plebanski tensor \cite{Pleb, Mc}, the Riemann-Cartan tensor \cite{Cartan1, Cartan2, Cartan3}, among others. 

The fact that $(2,2)$ double forms play a prominent role in general relativity by no means restrict their scope. In the so-called pre-metric approach to electrodynamics, their densitized versions play the role of the constitutive tensors governing the behavior of all sorts of local and linear continuous media \cite{Post, Obukhov1, Hehl1,Hehl2}. In area metric theory they appear as the building blocks in the generalization of pseudo-Riemannian geometry. Such generalizations naturally emerge as classical backgrounds in quantum string theory and quantum gauge theory finding a wealth of applications such as in the low energy action for $D$-branes and in the propagation of back-reacting photons in quantum electrodynamics \cite{Schuller1,Schuller2,Schuller3}. Furthermore, such tensors constitute the key ingredients in constructing other important tensorial objects such as the Bel-Robinson tensor, the Kummer tensor and the Rubilar-Tamm tensor \cite{quatrieme1, quatrieme2, Kummer}.

Given the variety of applications, it is not a surprise that the algebraic structure of these tensors have been studied by physicists and mathematicians along the years. However, by far, the main efforts have been dedicated to the particular case of curvaturelike double forms (the Riemann tensor itself is an instance of such objects). As is well known, the latter are characterized by the skew-symmetry property together with the interchange symmetry property and the first algebraic Bianchi identity. For this case, a wealth of results already constitute the basic toolkit of the relativist. Well known examples in four dimensions include the Ricci decomposition of the Riemann tensor \cite{gravitation, Hawking,Wald,Stephani}, the Ruse-Lanczos identity \cite{Lanczos1,Lanczos2, Ruse}, the so-called Lovelock identities \cite{Lovelock, Edgar1, Edgar2} and the Bel-Matte decomposition into lower-rank spacelike tensors \cite{Matte, Bel1,Bel2}. However, to the best of our knowledge, one hardly encounters self-contained literature dealing with similar results in the context of generic $(2,2)$ double forms. The aim of our work is to partially fulfill this gap. Although some of the results presented here do appear elsewhere \cite{Senovilla1, Senovilla2}, we hope that our schematic treatment is better-adapted to the study of the following open problems:

\begin{itemize}
    \item{Covariant hyperbolizations of electrodynamics inside linear and local media;}
    \item{Well-posedness of the Cauchy problem in lower-dimensional electrodynamics;}
    \item{Alternative derivations of the Tamm-Rubilar tensor using traditional methods;}
    \item{Algebraic classification and identification of invariants for constitutive tensors;}
\end{itemize}

Throughout, we adopt the following conventions. Spacetime indices $a,b,c,...$ take the values from $0$ to $3$; the greek letters $\alpha,\beta,\gamma,...$ serve to denote collective skew indices and take values from $1$ to $6$. Symbols inside curly brackets, such as $\{\Omega\}$ or $\{\Upsilon\}$, are reserved to describe an arbitrary collection of covariant and contravariant indices. Symmetrization and anti symmetrization are indicated as usual by round and square brackets respectively, e.g.
\begin{equation}
A_{(ab)}\equiv A_{ab}+A_{ba},\quad\quad\quad A_{[ab]}\equiv A_{ab}-A_{ba}.
\end{equation}



\section{Mathematical setting}

 \subsection{Generalities}


To begin with, we let $(M,g)$ denote a $m$-dimensional manifold with metric signature $(p,q)$, where $p+q=m$. Since our considerations will be essentially algebraic, we shall fix a point $x\in M$ and denote its tangent space by $V$, for simplicity. The vector space of all exterior $r$-forms, denoted by $\Lambda^{r}V^{*}$, is called the $r$-th exterior power of the cotangent space and its dimension reads as \cite{Lee, Bott} 
\begin{equation}
 \mbox{dim}(\Lambda^{r}V^{*})=\begin{pmatrix}
m \\
r \\
\end{pmatrix}=\frac{m!}{r!(m-r)!},\quad\quad\quad 0\leq r\leq m.
\end{equation}
The direct sum of all possible exterior powers, denoted by
\begin{equation}
\Lambda V^{*}\equiv\bigoplus_{r}\Lambda^{r}V^{*},
\end{equation}
is itself a vector space of dimension $2^{m}$ and the wedge product 
\begin{equation}
\wedge:\Lambda^{r}V^{*}\times\Lambda^{s}V^{*}\rightarrow\Lambda^{r+s}V^{*}
\end{equation}
turns $\Lambda V^{*}$ into an associative, anticommutative, graded algebra, called the \textit{exterior algebra} of $V^{*}$. The spaces we are interested in emerge from the tensor product of two such exterior algebras 
\begin{equation}
\mathcal{D}\equiv\Lambda V^{*}\otimes\Lambda V^{*}=\bigoplus_{0\leq r,s \leq m}\ \mathcal{D}^{r,s}
\end{equation}
where $\mathcal{D}^{r,s}\equiv\Lambda^{r}V^{*}\otimes \Lambda^{s}V^{*}$ is called the space of $(r,s)$ \textit{double forms}. An element of $\mathcal{D}^{r,s}$ is simply a $r+s$ covariant tensor which can be identified canonically with a bilinear form $\Lambda^{r}V\times\Lambda^{s}V\rightarrow\mathbb{R}$ and, if we stick to a (fixed) coordinate basis for $V$, an element $X\in\mathcal{D}^{r,s}$ is expanded as
\begin{equation}
X=\frac{1}{r!s!}X_{a_{1}...a_{r}b_{1}...b_{s}}\ dx^{a_{1}}\wedge...\wedge dx^{a_{r}} \otimes dx^{b_{1}}\wedge...\wedge dx^{b_{s}}.
\end{equation}
Hence, we have two blocks of skew-symmetric indices (see Appendix A for conventions)
\begin{equation}
X_{a_{1}...a_{r}b_{1}...b_{s}}=\frac{1}{r!}X_{[a_{1}...a_{r}]b_{1}...b_{s}}=\frac{1}{s!}X_{a_{1}...a_{r}[b_{1}...b_{s}]}.  
\end{equation}
The exterior product of two double forms is defined as
\begin{equation}
X\KN Y\equiv(\theta_{1}\wedge\omega_{1})\otimes(\theta_{2}\wedge\omega_{2})\in\mathcal{D}^{r+k,\ s+l}    \end{equation}
where $X=\theta_{1}\otimes\theta_{2}\in\mathcal{D}^{r,s}$ and $Y=\omega_{1}\otimes\omega_{2}\in\mathcal{D}^{k,l}$. In abstract index notation, this product reads as
\begin{equation}\label{KNproduct}
(X\KN Y)_{a_{1}...a_{r}b_{1}...b_{k}c_{1}...c_{s}d_{1}...d_{l}}=X_{[a_{1}...a_{r}[c_{1}...c_{s}}Y_{b_{1}...b_{k}]d_{1}...d_{l}]}\end{equation}
where total antisymmetrization is supposed to act only on the $[ab]$ and $[cd]$ blocks independently. An important instance of this product is when both the double forms are identified with rank two tensors. In this case if $k,h\in\mathcal{D}^{1,1}$ Eq. (\ref{KNproduct}) reduces to
\begin{equation}
(k\KN h)_{abcd}=k_{[a[c}h_{b]d]}=k_{ac}h_{bd}-k_{ad}h_{bc}+h_{ac}k_{bd}-h_{ad}k_{bc},
\end{equation}
in which it is called a \textit{Kulkarni-Nomizu product}. More generally, the exterior product $\KN$ turns $\mathcal{D}$ into the bi-graded associative algebra studied by mathematicians such as Greub, Thorpe and Kulkarni in the late sixties \cite{Greub, Thorpe,Kulkarni}. The reader is invited to consult \cite{Labbi1, Labbi2}, where several classical results concerning the algebra of multilinear forms are obtained within this framework. 

\subsection{The space of $(2,2)$ double forms}

This paper deals with the particular case of $(2,2)$ double forms in a four-dimensional spacetime. However, we start by recalling some general properties of these tensors in a space of arbitrary dimension and signature. In a fixed coordinate basis, the components of an element of $\mathcal{D}^{2,2}$ are characterized by two blocks of skew-symmetric indices
 \begin{equation}
  X_{abcd}=-X_{bacd},\quad\quad\quad X_{abcd}=-X_{abdc}.
 \end{equation}
Such a tensor has a total of $m^{2}(m-1)^{2}/4$ independent components and  a simple combinatorial exercise shows that we can construct at most six linearly independent tensors using permutations of the indices. A convenient choice of them is provided by
\begin{eqnarray}
\ X_{abcd},\ X_{acdb},\ X_{adbc}\ ;\ X_{cdab},\ X_{dbac},\ X_{bcad}.
\end{eqnarray}
The double form obtained by the particular index mapping $ab\leftrightarrow cd$ deserves a special name. Due to obvious reasons, we shall call $X_{cdab}$ the transpose of $X_{abcd}$. We then notice that the first triple above is obtained from the original tensor by cyclic permutations of the indices $bcd$, whereas the second triple is obtained from the transpose by similar permutations. 

Apart from the skew symmetries of each of the blocks one may have to deal with some additional symmetries involving indices of different blocks and an elementary calculation originally due to Schouten shows that (see \cite{Schouten} for details)
\begin{equation}\label{Schouten}
X_{abcd}=X_{cdab}+\frac{1}{4}(X_{a[bcd]}-X_{b[cda]}-X_{c[dab]}+X_{d[abc]}).
\end{equation}
In most of physical applications we are interested in double forms which may carry special algebraic symmetries and simple examples of the latter already appear in the trivial decomposition
\begin{equation}\label{SA}
X_{abcd}\ =\ ^{(s)}X_{abcd}\ +\ ^{(a)}X_{abcd},
\end{equation}
where
\begin{eqnarray}\label{sym1}
^{(s)}X_{abcd}&\equiv&\frac{1}{2}(X_{abcd}+X_{cdab}),\quad\quad\quad ^{(s)}X_{abcd}=+^{(s)}X_{cdab},\\\label{sym2}
^{(a)}X_{abcd}&\equiv&\frac{1}{2}(X_{abcd}-X_{cdab}),\quad\quad\quad ^{(a)}X_{abcd}=-^{(a)}X_{cdab}.
\end{eqnarray}
Denoting the symmetric and antisymmetric subspaces, respectively, by $\mbox{Sym}$ and $\mbox{Skew}$, there follows the orthogonal direct sum
\begin{equation}\label{ordirect1}
\mathcal{D}^{2,2}= \mbox{Sym}\oplus\mbox{Skew},
 \end{equation}
with
 \begin{equation}
 \mbox{dim}(\mbox{Sym})=\frac{1}{8}(m^{2}-m)(m^{2}-m+2),\quad\quad\quad  \mbox{dim}(\mbox{Skew})=\frac{1}{8}(m+1)m(m-1)(m-2),
 \end{equation}
and Eq. (\ref{Schouten}) gives 
\begin{equation}
^{(s)}X_{a[bcd]}\ -\ ^{(s)}X_{b[cda]}\ =\ ^{(s)}X_{c[dab]}\ -\ ^{(s)}X_{d[abc]},
\end{equation}
\begin{equation}
^{(a)}X_{abcd}\ =\ \frac{1}{8}\left(\ ^{(a)}X_{a[bcd]}\ -\ ^{(a)}X_{b[cda]}\ -\ ^{(a)}X_{c[dab]}\ +\ ^{(a)}X_{d[abc]}\right).  
\end{equation}
Concerning total antisymmetrization, an elementary calculation shows that
\begin{equation}
^{(s)}X_{[abcd]}=4\ ^{(s)}X_{a[bcd]},\quad\quad\quad ^{(a)}X_{[abcd]}=0.
\end{equation}
For obvious reasons, $^{(s)}X_{abcd}$ is called a \textit{symmetric} $(2,2)$ \textit{double form} whereas $^{(a)}X_{abcd}$ is called an \textit{antisymmetric} $(2,2)$ \textit{double form} and the above decomposition is orthogonal in the sense that
\begin{equation}
X_{abcd}X^{abcd}\ =\ ^{(s)}X_{abcd}\ ^{(s)}X^{abcd}+\ ^{(a)}X_{abcd}\ ^{(a)}X^{abcd}.
\end{equation}

\subsection{Invariant subspaces}

When the group of general linear transformations acts on the tangent space, it induces a mapping of the space of double forms onto itself and, since all symmetries discussed above are preserved by this mapping, we say that $\mbox{Sym}$ and $\mbox{Skew}$ are invariant subspaces under the group action. A refined version of Eq. (\ref{SA}) containing three invariant subspaces under the action of $GL(m,\mathbb{R})$ is the following
\begin{equation}\label{ordirect2}
\mathcal{D}^{2,2}=\underbrace{\mbox{Curv}\ \oplus\ \Lambda^{4}V^{*}}_{\mbox{Sym}}\ \oplus\ \mbox{Skew}.
 \end{equation}
where $\mbox{Curv}$ is the orthogonal complement of the space of $4$-forms in the symmetric space. This subset is called the vector space of \textit{algebraic curvature tensors} and an element of the latter will be called \textit{curvaturelike} for short. Its dimension reads as
\begin{equation}
 \mbox{dim}(\mbox{Curv})=\frac{m^{2}(m^{2}-1)}{12}
 \end{equation}
and, in practice, we have
\begin{equation}\label{CLA}
 X_{abcd}\ =\ ^{(c)}X_{abcd}\ +\ ^{(4)}X_{abcd}\ +\ ^{(a)}X_{abcd},
 \end{equation}\\
 where $^{(c)}X_{abcd}$ is curvaturelike and $^{(4)}X_{abcd}$ is a $4$-form. These tensors are distinguished within the symmetric subspace by the following additional symmetries\\
 \begin{eqnarray}\label{sym3}
^{(c)}X_{a[bcd]}=0,\quad\quad\quad ^{(4)}X_{[abcd]}=4!\ ^{(4)}X_{abcd},\quad\quad\quad ^{(c)}X_{abcd}\ ^{(4)}X^{abcd}=0.
\end{eqnarray}\\
More generally, it can be checked that an arbitrary double form satisfying the algebraic Bianchi identity, $X_{a[bcd]}=0$, is necessarily symmetric and the so-called Milnor octahedron construction gives a beautiful geometric interpretation for this argument \cite{Milnor}. 

So far, the metric tensor did not play any role in the above decompositions. In order to involve it explicitly, we start by associating to a generic double form a second order tensor and a scalar, respectively defined by
\begin{equation}
X_{ac}\equiv g^{bd}X_{abcd},\quad\quad\quad X\equiv g^{ac}X_{ac}.
\end{equation}
Henceforth, we call $X_{ac}$ the \textit{first trace} and $X$ the \textit{second trace} of the double form, for simplicity\footnote{ \samepage Mathematicians often refer to these operations as the Ricci contraction maps. If $X_{abcd}$ is identified with the Riemann curvature tensor, the first contraction would be the Ricci tensor whereas the second contraction would be the scalar curvature.}. Furthermore, removing the second trace from the first, we obtain the traceless \textit{hatted tensor}
\begin{equation}\label{tlt}
\hat{X}_{ab}\equiv X_{ab}-\frac{1}{m}Xg_{ab},\quad\quad\quad\hat{X}=0.
\end{equation}
We shall say that the double form itself is \textit{traceless} whenever any possible contraction of its indices vanish and \textit{hatless} whenever its associated rank-2 hatted tensor vanishes. According to these definitions, Eq. (\ref{CLA}) gives
\begin{equation}
X_{ac}\ =\ ^{(c)}X_{ac}\ +\ ^{(a)}X_{ac},\quad\quad\quad X\ =\ ^{(c)}X,\quad\quad\quad\hat{X}_{ac}\ =\ ^{(c)}\hat{X}_{ac}\ +\ ^{(a)}\hat{X}_{ac}, 
\end{equation}
since the four-form term has no associated traces. The main reason behind these definitions is the fact that for $m\geq 3$, every $(2,2)$ double form admit a \textit{trace decomposition} as follows
\begin{equation}\label{Tdec}
X_{abcd}=W_{abcd}+\frac{1}{(m-2)}g_{[a[c}\hat{X}_{b]d]}+\frac{X}{m(m-1)}g_{abcd},
\end{equation}
where, again, antisymmetrization is supposed to act only on the pairs $[ab]$ and $[cd]$ individually and $g_{abcd}=g_{a[c}g_{bd]}$ is sometimes called the bi-metric. The above decomposition reinforces the naturalness of the Kulkarni-Nomizu products in the context of double forms (see Appendix B). Here, the first tensor on the r.h.s. is traceless, the second has only the first trace and the third has both the first and second traces. 

It turns out that the trace decomposition provides a splitting of the space algebraic curvature tensors into three smaller spaces and the splitting of the antisymmetric space into two new subspaces. Indeed, applying Eq. (\ref{Tdec}) to each term in Eq. (\ref{CLA}) yields the breaking up of the subspaces as follows
\begin{equation}
^{(c)}X_{abcd}\ =\ ^{(1)}X_{abcd}\ +\ ^{(2)}X_{abcd}\ + ^{(3)}X_{abcd},
\end{equation}
\begin{equation}
^{(a)}X_{abcd}\ =\ ^{(5)}X_{abcd}\ +\ ^{(6)}X_{abcd},
\end{equation}
where $^{(1)}X_{abcd}$ is curvaturelike traceless, $^{(5)}X_{abcd}$ is antisymmetric traceless and the Kulkarni-Nomizu products involving the metric and the hatted tensors read as
\begin{equation}
^{(2)}X_{abcd}\ \equiv\  \frac{1}{(m-2)}g_{[a[c}\ ^{(c)}\hat{X}_{b]d]},\quad\quad ^{(3)}X_{abcd}\equiv \frac{X}{m(m-1)}g_{abcd}\quad\quad ^{(6)}X_{abcd}\ \equiv\  \frac{1}{(m-2)}g_{[a[c}\ ^{(a)}\hat{X}_{b]d]}.
\end{equation}
Collecting all the above pieces, and noticing that the space of four-forms remain intact, one concludes that an arbitrary double form decomposes as the sum of six \textit{irreducible} terms
 \begin{equation}\label{Irreddec}
X_{abcd}=\ ^{(1)}X_{abcd}\ +\ ^{(2)}X_{abcd}\ + ^{(3)}X_{abcd}\ +\ ^{(4)}X_{abcd}+\ ^{(5)}X_{abcd}\ +\ ^{(6)}X_{abcd}.
\end{equation}
In what follows, we shall represent the corresponding invariant subspaces as\\
\begin{equation}\label{almostfull}
\mathcal{D}^{2,2}=^{(s)}\mbox{Weyl}\oplus^{(s)}\mbox{KN1}\oplus^{(s)}\mbox{KN2}\oplus^{(s)}\Lambda^{4}V^{*}\oplus^{(a)}\mbox{Weyl}\oplus^{(a)}\mbox{KN1},
\end{equation}\\
and stick to the following terminology: $^{(s)}\mbox{Weyl}$ is the space of \textit{symmetric Weyl tensors}, $^{(s)}\mbox{KN1}$ is the space of \textit{symmetric Kulkarni-Nomizu of the first kind}, $^{(s)}\mbox{KN2}$ the space of \textit{symmetric Kulkarni-Nomizu of the second kind} and similarly for the antisymmetric subspaces. The first three symmetric spaces on the r.h.s. of Eq. (\ref{almostfull}) are well known to relativists and their direct sum is often referred as the Ricci decomposition. Little-known, however, is the fact that the antisymmetric part also admits a similar decomposition in terms of an antisymmetric traceless tensor and an antisymmetric Kulkarni-Nomizu product. For $m\geq3$, decomposition Eq. (\ref{almostfull}) is irreducible under the pseudo-orthogonal group $O(p,q)$ and the algebraic properties of all relevant tensors are summarized in the following table\footnote{For additional details concerning the uniqueness, irreducibility and invariance of decompositions, we refer the reader to \cite{Riemann-Cartan,Obukhov1, Itin}.}:\\

\begin{table}[h]
\begin{center}
\begin{tabular}{|c|c|c|r|r|r|r|r|c|}
\hline
\multicolumn{2}{|c|}{Subspace} & Tensor & $[abcd]$ & $a[bcd]$ & 1st. & 2nd. & Hat & dimensionality  \\
\hline\hline
\multirow{3}{*}{\mbox{Curv}} & $^{(s)}\mbox{Weyl}$ & $^{(1)}X_{abcd}$ & $0$ & $0$ & $0$ & $0$ &$0$& $\tfrac{1}{12}m(m+1)(m+2)(m-3)$\\\cline{2-9}
 & $^{(s)}\mbox{KN1}$ & $^{(2)}X_{abcd}$ & $0$ & $0$ & $\neq 0$ & $0$ &$\neq 0$ & $\tfrac{1}{2}(m-1)(m+2)$\\\cline{2-9}
 & $^{(s)}\mbox{KN2}$ & $^{(3)}X_{abcd}$ & $0$ & $0$ & $\neq 0$ & $\neq 0$ & $0$ & $1$\\\hline
$\Lambda^4V^*$ & $^{(s)}\Lambda^4V^*$ & $^{(4)}X_{abcd}$ & $\neq 0$ & $\neq 0$ & $0$ & $0$ & $0$ & $\binom{m}{4}$\\\hline
\multirow{2}{*}{\mbox{Skew}} & $^{(a)}\mbox{Weyl}$ & $^{(5)}X_{abcd}$ & $0$ & $\neq 0$ & $0$ & $0$ & $0$& $\tfrac{1}{8}m(m-1)(m+2)(m-3)$\\\cline{2-9}
 & $^{(a)}\mbox{KN2}$ & $^{(6)}X_{abcd}$ & $0$ & $\neq 0$ & $\neq 0$ & $0$ &$\neq 0$ & $\tfrac{1}{2}m(m-1)$\\\hline
\end{tabular}
\caption{\label{table1}Invariant subspaces of the space of $(2,2)$ double forms; the possible additional symmetries and traces of its tensors, and the corresponding dimension of each subspace.}
\end{center}
\end{table}

\section{Hodge duality and the Ruse-Lanczos identity}

From now on we restrict ourselves to a four-dimensional spacetime $(M,\ g)$ with signature convention $(+,-,-,-)$. In this case, Hodge dualization provides a simple way of producing new $(2,2)$ double forms from given ones. In particular, we construct the left, right and double duals of $X\in\mathcal{D}^{2,2}$ as follows

 \begin{equation}\label{doubledual}
\overleftarrow{\star} X_{abcd}\equiv \frac{1}{2}\varepsilon_{ab}^{\phantom a\phantom apq}X_{pqcd},\quad\quad\quad \overrightarrow{\star}X_{abcd} \equiv \frac{1}{2}\varepsilon_{cd}^{\phantom a\phantom apq}X_{abpq},\quad\quad \overleftrightarrow{\star} X_{abcd} \equiv \frac{1}{4}\varepsilon_{ab}^{\phantom a\phantom a pq}\varepsilon_{cd}^{\phantom a\phantom a rs}X_{pqrs}.
\end{equation}
According to this notation, the direction of the arrow on top of the $\star$ indicates the pair of skew indices which are to be dualized. This convention is not standard. However, it has the advantage of avoiding an excessive number of stars in expressions involving several dualized double forms. A direct consequence of the hyperbolic signature of the metric is that
\begin{equation}
\overleftarrow{\star}\circ\overleftarrow{\star}=-1,\quad\quad\quad \overrightarrow{\star}\circ\overrightarrow{\star}=-1,\quad\quad\quad \overleftrightarrow{\star}\circ\overleftrightarrow{\star}=1.
\end{equation}
Let us now prove a fundamental identity concerning the Hodge duals associated to an arbitrary $(2,2)$ double form. A prototype of the identity was first obtained by Einstein, in particular coordinates, while studying the splitting of the Riemann curvature tensor into its self-dual and anti-self-dual parts \cite{Einstein}.  The same relation was later obtained, in a more covariant fashion, by Lanczos in the context of scale-invariant action principles \cite{Lanczos1,Lanczos2} and by Ruse in the so-called line-geometry of the Riemann tensor \cite{Ruse}. 

A direct generalisation of the identity starts with the definition of the double dual in its mixed form and replacing the product of Levi-Civita tensors by the corresponding generalized Kronecker delta (see Appendix A)
\begin{equation}\label{R-L}
\overleftrightarrow{\star} X^{ab}_{\phantom a\phantom a cd}=\frac{1}{4}\varepsilon^{abpq}\varepsilon_{cdrs}X_{pq}^{\phantom a\phantom a rs}=-\frac{1}{4}\delta^{abpq}_{\phantom a\phantom a\phantom a\phantom a cdrs}X_{pq}^{\phantom a\phantom a rs}.
\end{equation}
The last term is expanded as
\begin{eqnarray*}
-\frac{1}{4}(\delta^{a}_{\phantom a c}\delta^{bpq}_{\phantom a\phantom a\phantom a drs}-\delta^{a}_{\phantom a d}\delta^{bpq}_{\phantom a\phantom a\phantom a rsc}+\delta^{a}_{\phantom a r}\delta^{bpq}_{\phantom a\phantom a\phantom a scd}-\delta^{a}_{\phantom a s}\delta^{bpq}_{\phantom a\phantom a\phantom a cdr})X_{pq}^{\phantom a\phantom a rs}.
\end{eqnarray*}
\noindent Calculating by parts, there follow
\begin{itemize}
\item{$\delta^{a}_{\phantom a c}\delta^{bpq}_{\phantom a\phantom a\phantom a drs}X_{pq}^{\phantom a\phantom a rs}=2\delta^{a}_{\phantom a c}\delta^{b}_{\phantom a d}X-4\delta^{a}_{\phantom a c}X_{d}^{\phantom a b}$,}
\item{$\delta^{a}_{\phantom a d}\delta^{bpq}_{\phantom a\phantom a\phantom a rsc}X_{pq}^{\phantom a\phantom a rs}=2\delta^{a}_{\phantom a d}\delta^{b}_{\phantom a c}X-4\delta^{a}_{\phantom a d}X_{c}^{\phantom a b}$,}
\item{$\delta^{a}_{\phantom a r}\delta^{bpq}_{\phantom a\phantom a\phantom a scd}X_{pq}^{\phantom a\phantom a rs}=2\delta^{b}_{\phantom a c}X_{d}^{\phantom a a}-2\delta^{b}_{\phantom a d}X_{c}^{\phantom a a}+2X_{cd}^{\phantom a\phantom a ab}$,}
\item{$\delta^{a}_{\phantom a s}\delta^{bpq}_{\phantom a\phantom a\phantom a cdr}X_{pq}^{\phantom a\phantom a rs}=2\delta^{b}_{\phantom a d}X_{c}^{\phantom a a}-2\delta^{b}_{\phantom a c}X_{d}^{\phantom a a}-2X_{cd}^{\phantom a\phantom a ab}$,}
\end{itemize}
and collecting the terms we obtain, after simple manipulations
\begin{equation}\label{Ruse-Lanczos}
\overleftrightarrow{\star} X_{abcd}+X_{cdab}=(g_{ac}X_{db}-g_{ad}X_{cb}+g_{bd}X_{ca}-g_{bc}X_{da})-\frac{1}{2}Xg_{abcd}.
\end{equation}
Interestingly, the right hand side of the equation may be written in terms of a Kulkarni-Nomizu product of the metric $g_{ab}$ and the hatted tensor $\hat{X}_{ab}$ associated to $X_{abcd}$. Indeed, using the definition of the trace-free tensor Eq. (\ref{tlt}), we obtain the following proposition
\begin{prop}[Ruse-Lanczos identity]
Let $X_{abcd}$ denote an arbitrary element of the vector space $\mathcal{D}^{2,2}$. Then, the following relation holds
\begin{equation}\label{HD}
\overleftrightarrow{\star}X_{abcd}+X_{cdab}=g_{[a[c}\hat{X}_{d]b]},
\end{equation}
where $\hat{X}_{ab}= X_{ab}-\frac{1}{4}Xg_{ab}$ and anti-symmetrization is supposed to act on the pairs $[ab]$ and $[cd]$ individually.
\end{prop}
\noindent Equation (\ref{HD}) is a generalised version of the classical Ruse-Lanczos identity \cite{Einstein, Lanczos1,Lanczos2,Ruse}. The main difference between the two approaches resides in the particular position of the indices they generate: here, we are forced to consider the transposes of both the double form and the hatted tensor. This is expected, since we are not assuming any additional algebraic symmetries, as is the case for algebraic curvature tensors studied in the above-mentioned articles. Roughly speaking, the identity says that the double dual of any double form plus its transpose generates a new double form possessing only the first trace: in other words, the linear combination lies in the subspace defined by the orthogonal direct sum $^{(s)}\mbox{KN1}\oplus\ ^{(a)}\mbox{KN1}$. Indeed, a direct consequence of the above proposition is the following corollary

\begin{cor}
Decompose a generic double form into its symmetric and anti-symmetric parts as $X_{abcd}=S_{abcd}+A_{abcd}$. Then, applying Eq. (\ref{HD}) to each invariant part gives
\begin{eqnarray}
\ \overleftrightarrow{\star}S_{abcd}+S_{abcd}=g_{[a[c}\hat{S}_{b]d]},\quad\quad\quad
\ \overleftrightarrow{\star}A_{abcd}-A_{abcd}=-g_{[a[c}\hat{A}_{b]d]}.
\end{eqnarray}
\end{cor}
\noindent We notice that all indices are now in their conventional form. Also, one immediately sees that the first combination belongs to $^{(s)}\mbox{KN1}$, whereas the second belongs to $^{(a)}\mbox{KN1}$. Furthermore, since $\hat{S}_{ab}$ and $\hat{A}_{ab}$ are trace-less by construction, there follows the first trace relations

\begin{cor}
Contracting $b$ with $d$ above and rearranging the terms, gives
\begin{equation}\label{HDT}
\overleftrightarrow{\star} S_{ac}=S_{ac}-\frac{1}{2}Sg_{ac},\quad\quad\quad \overleftrightarrow{\star} A_{ac}=-A_{ac}.
\end{equation}
\end{cor}
\noindent We notice that the first relation above reduces to the definition of the well-known Einstein tensor in the particular case $S_{abcd}\mapsto R_{abcd}$. Finally, contracting Eqs. (\ref{HDT}) gives us the second trace relation $\overleftrightarrow{\star}S=-S$. In particular, the latter shows us that not all possible invariants (scalars) constructed from a double form and its Hodge duals are independent. 

It is instructive to see how the generalised Ruse-Lanczos identity interacts with the irreducible decomposition provided by Eq. (\ref{Irreddec}). To start with, we notice that, in a four-dimensional spacetime, the invariant subspaces described in Section II have the following corresponding dimensions
\begin{equation}\label{irreducfour}
X_{abcd}=\underbrace{^{(1)}X_{abcd}}_{10}+\underbrace{^{(2)}X_{abcd}}_{9}+\underbrace{^{(3)}X_{abcd}}_{1}+\underbrace{^{(4)}X_{abcd}}_{1}+\underbrace{^{(5)}X_{abcd}}_{9}+\underbrace{^{(6)}X_{abcd}}_{6}.
\end{equation}
By applying Eq. (\ref{HD}) to each tensor above, we have
\begin{eqnarray}\label{RL1}
&&\overleftrightarrow{\star}\ ^{(1)}X_{abcd}=-^{(1)}X_{abcd},\\
&&\overleftrightarrow{\star}\ ^{(2)}X_{abcd}=+^{(2)}X_{abcd},\\
&&\overleftrightarrow{\star}\ ^{(3)}X_{abcd}=-^{(3)}X_{abcd},\\
&&\overleftrightarrow{\star}\ ^{(4)}X_{abcd}=-^{(4)}X_{abcd},\\
&&\overleftrightarrow{\star}\ ^{(5)}X_{abcd}=+^{(5)}X_{abcd},\\\label{RL2}
&&\overleftrightarrow{\star}\ ^{(6)}X_{abcd}=-^{(6)}X_{abcd},
\end{eqnarray}
which shows that the double dual, when restricted to each subspace of $\mathcal{D}^{2,2}$, works as an automorphism that acts either as the identity or a reflection about the origin. In particular, this means that the images have precisely the same symmetries as the corresponding domains. 

A $(2,2)$ double form is said to be \textit{self-dual} if it is an eigentensor of the map $\overleftrightarrow{\star}$ with eigenvalue $+1$ and \textit{anti-self-dual} if it is an eigentensor with eigenvalue $-1$. In the early days of general relativity, Rainich have proposed a splitting of the Riemann curvature tensor into two parts as follows \cite{Rainich}
\begin{equation}
X_{abcd}=\underbrace{\frac{1}{2}(X_{abcd}+\overleftrightarrow{\star} X_{abcd})}_{self-dual}+\underbrace{\frac{1}{2}(X_{abcd}-\overleftrightarrow{\star}X_{abcd})}_{anti-self-dual}.    
\end{equation}
This decomposition is at the roots of the so-called Rainich-Wheeler-Misner program for the unification of gravity with electromagnetism \cite{Wheeler} and, using Eqs. (\ref{RL1})-(\ref{RL2}) it is straightforward to show that each part above reads, respectively, as
\begin{equation}
\underbrace{^{(2)}X_{abcd}\ +\ ^{(5)}X_{abcd}}_{18},
\end{equation}
\begin{equation}
\underbrace{^{(1)}X_{abcd}\ +\ ^{(3)}X_{abcd}\ +\ ^{(4)}X_{abcd}\ +\ ^{(6)}X_{abcd}}_{18}.
\end{equation}
If $X_{abcd}$ is identified with the Riemann curvature tensor $R_{abcd}$, the self-dual part reduces to a total of $9$ independent components (traceless part of Ricci tensor) while the anti-self-dual part consists of the remaining $11$ (Weyl tensor plus scalar curvature).

We now turn our attention to the action of left and right duals on the space of double forms. This analysis is a bit more involved since the corresponding maps somehow mix the subspaces. To see this, we start by noticing that from Eqs. (\ref{RL1})-(\ref{RL2}), one has
\begin{eqnarray}\label{RL3}
&&\overleftarrow{\star}\ ^{(1)}X_{abcd}=+\overrightarrow{\star}\ ^{(1)}X_{abcd},\\
&&\overleftarrow{\star}\ ^{(2)}X_{abcd}=-\overrightarrow{\star}\ ^{(2)}X_{abcd},\\
&&\overleftarrow{\star}\ ^{(3)}X_{abcd}=+\overrightarrow{\star}\ ^{(3)}X_{abcd},\\
&&\overleftarrow{\star}\ ^{(4)}X_{abcd}=+\overrightarrow{\star}\ ^{(4)}X_{abcd},\\
&&\overleftarrow{\star}\ ^{(5)}X_{abcd}=-\overrightarrow{\star}\ ^{(5)}X_{abcd},\\\label{RL4}
&&\overleftarrow{\star}\ ^{(6)}X_{abcd}=+\overrightarrow{\star}\ ^{(6)}X_{abcd}.
\end{eqnarray}
The above relations imply that the target spaces of the left and right duals coincide, when restricted to each invariant subspace. Therefore, it suffices to look at the behavior of the map $\overleftarrow{\star}$. By carefully inspecting the algebraic symmetries of the l.h.s. above, we conclude that the subspaces of $\mbox{Sym}$ are mapped as follows\\
\begin{eqnarray}
&&\overleftarrow{\star}:\ ^{(s)}\mbox{Weyl}\ \rightarrow\   ^{(s)}\mbox{Weyl},\quad\quad\overleftarrow{\star}:\ ^{(s)}\mbox{KN1}\ \rightarrow\ ^{(a)}\mbox{Weyl},
\end{eqnarray}
\begin{eqnarray}
\overleftarrow{\star}:\ ^{(s)}\mbox{KN2}\ \rightarrow\ ^{(s)}\Lambda^{4}V^{*},\quad\quad \overleftarrow{\star}:\  ^{(s)}\Lambda^{4}V^{*}\ \rightarrow\ ^{(s)}\mbox{KN2},
\end{eqnarray}
\\
whereas the subspaces of $\mbox{Skew}$ have the following targets\\
\begin{eqnarray}
&&\overleftarrow{\star}:\ ^{(a)}\mbox{Weyl}\ \rightarrow\  ^{(s)}\mbox{KN1},\quad\quad\overleftarrow{\star}:\ ^{(a)}\mbox{KN1}\rightarrow\ ^{(a)}\mbox{KN1}.
\end{eqnarray}
The latter shows that, in general, the symmetries of the pieces are not preserved by the action of single Hodge dual. An interesting by-product of the above mappings is the unexpected correspondence between $^{(s)}\mbox{KN1}$ and $^{(a)}\mbox{Weyl}$ via duality. Indeed, since $\star$ creates an isomorphism between these subspaces, there exists a symmetric-traceless tensor $\hat{M}_{ab}$, such that
\begin{equation}
^{(a)}W_{abcd}=\overleftarrow{\star}\ g_{[a[c}\hat{M}_{b]d]}=\varepsilon_{abc}^{\phantom a\phantom a\phantom a q}\hat{M}_{dq}-\varepsilon_{abc}^{\phantom a\phantom a\phantom a q}\hat{M}_{dq},\quad\quad\quad \hat{M}_{ac}=\frac{1}{2}\overleftarrow{\star}\ ^{(a)}W_{ac}.
\end{equation}
Similar parametrizations have been derived in \cite{Obukhov1, Itin} and find its main applications in the study of skewon modifications to the Minkowski vacuum. We refer the reader to \cite{Senovilla1} where a detailed treatment of the Ruse-Lanczos identity in arbitrary dimension is given within the framework of $r$-fold forms.



\section{Bel-Matte decomposition}

We now investigate the splitting of a generic $(2,2)$ double form $X_{abcd}$ induced by a given vector $t^{q}$ at a spacetime point. A particular case of the splitting is originally due to Matte \cite{Matte} and Bel \cite{Bel1, Bel2} and gives rise to a set of lower order tensors with properties quite similar to the electric and magnetic fields. In the restricted case of the Riemann tensor, the set is to be identified with the so-called electrogravitic, magnetogravitic and topogravitic tensors. In order to obtain a generalised version of the Bel-Matte decomposition, we first consider a tensor of arbitrary rank, such that
\begin{equation}
Y^{\{\Omega\}}_{\phantom a\phantom a\phantom a\phantom a ab}=-Y^{\{\Omega\}}_{\phantom a\phantom a\phantom a\phantom a ba},
\end{equation}
with $\{\Omega\}$ schematically denoting an arbitrary number of additional lower and/or upper indices.  If we are given a vector of arbitrary character $t^{q}\in V$, we may construct the following projected tensors\footnote{Throughout, we use capital gothic letters such as $\mathfrak{A}$,$\mathfrak{B}$,$\mathfrak{C}$,$\mathfrak{D}$ to denote contractions of a tensor with a given vector.}
\begin{equation}
\mathfrak{E}^{\{\Omega\}}_{\phantom a\phantom a\phantom a\phantom a a}\equiv Y^{\{\Omega\}}_{\phantom a\phantom a\phantom a\phantom a ab}t^{b},\quad\quad\quad \mathfrak{H}^{\{\Omega\}}_{\phantom a\phantom a\phantom a\phantom a a}\equiv -(Y^{\{\Omega\}}_{\phantom a\phantom a\phantom a\phantom a ab}\star) t^{b},
\end{equation}
and, due to the skew symmetry of the pair $[ab]$, these tensors are automatically transversal to the vector, in the sense that
\begin{equation}
\mathfrak{E}^{\{\Omega\}}_{\phantom a\phantom a\phantom a\phantom a a}t^{a}=0,\quad\quad\quad \mathfrak{H}^{\{\Omega\}}_{\phantom a\phantom a\phantom a\phantom a a}t^{a}=0.
\end{equation}

\begin{prop}[Bel-Matte decomposition]
\label{BMdec}
With the above conventions, there follows the algebraic identity of projections
\begin{equation}\label{EHSP}
(t_{q}t^{q})Y^{\{\Omega\}ab}=(g^{ab}_{\phantom a\phantom a cd}\mathfrak{E}^{\{\Omega\}c}+\varepsilon^{ab}_{\phantom a\phantom a cd}\mathfrak{H}^{\{\Omega\}c})t^{d}.
\end{equation}
\end{prop}
\begin{proof}
Absorbing the Levi-Civita tensor appropriately in the r.h.s gives
\begin{eqnarray*}
\varepsilon^{ab}_{\phantom a\phantom a cd}\mathfrak{H}^{\{\Omega\}c}t^{d}&=&-\varepsilon^{ab}_{\phantom a\phantom a cd}(Y^{\{\Omega\}cs}\star)t_{s}t^{d}\\
&=&-\frac{1}{2}\varepsilon^{abcd}Y^{\{\Omega\}lm}\varepsilon_{lmcs}t^{s}t_{d}\\
&=&+\frac{1}{2}\delta^{abd}_{\phantom a\phantom a\phantom a lms}Y^{\{\Omega\}lm}t^{s}t_{d}\\
&=&(t_{q}t^{q})Y^{\{\Omega\}ab}-\mathfrak{E}^{\{\Omega\}[a}t^{b]}.
\end{eqnarray*}
Writing the anti-symmetrization operation in the last term in terms of $g_{abcd}$ and rearranging the equation then gives the desired result. 
\end{proof}

 To obtain the generalised decomposition, we shall apply Eq. (\ref{EHSP}) two times. We start by identifying the first pair of the double form $X_{abcd}$ with the collective index $\{\Omega\}$ to write
\begin{equation}\label{first}
(t_{q}t^{q})X_{abcd}=(g_{cd}^{\phantom a\phantom a rs}\mathfrak{E}_{abr}+\varepsilon_{cd}^{\phantom a\phantom a rs}\mathfrak{H}_{abr})t_{s},\quad\quad\mathfrak{E}_{abr}=X_{abrl}t^{l},\quad\quad \mathfrak{H}_{abr}=-\overrightarrow{\star}X_{abrl} t^{l}.
\end{equation}
Now, by identifying the index $r$ with $\{\Omega\}$ in $\mathfrak{E}_{abr}$ and $\mathfrak{H}_{abr}$, there follows
\begin{equation}\label{second}
(t_{q}t^{q})\mathfrak{E}_{abr}=(g_{ab}^{\phantom a\phantom a pq}\mathfrak{A}_{pr}+\varepsilon_{ab}^{\phantom a\phantom a pq}\mathfrak{C}_{pr})t_{q},\quad\quad(t_{q}t^{q})\mathfrak{H}_{abr}=(g_{ab}^{\phantom a\phantom a pq}\mathfrak{B}_{pr}+\varepsilon_{ab}^{\phantom a\phantom a pq}\mathfrak{D}_{pr})t_{q},
\end{equation}
where
\begin{equation}\label{third}
\mathfrak{A}_{ac}\equiv\mathfrak{E}_{abc}t^{b},\quad \mathfrak{C}_{ac}\equiv -\overleftarrow{\star} \mathfrak{E}_{abc}t^{b},\quad\quad\quad\mathfrak{B}_{ac}\equiv \mathfrak{H}_{abc}t^{b},\quad \mathfrak{D}_{ac}\equiv -\overleftarrow{\star} \mathfrak{H}_{abc}t^{b}.
\end{equation}
Assuming $t_{q}t^{q}\neq 0$, and using Eqs. (\ref{second}) and (\ref{third}) into (\ref{first}), yields
\begin{equation}\label{Beldec}
X_{abcd}=\{g_{abpq}(g_{cdrs}\mathfrak{A}^{pr}+\varepsilon_{cdrs}\mathfrak{B}^{pr})+\varepsilon_{abpq}(g_{cdrs}\mathfrak{C}^{pr}+\varepsilon_{cdrs}\mathfrak{D}^{pr})\}t^{q}t^{s}/(t_{m}t^{m})^{2},
\end{equation}
with the $2$-index tensors given by
\begin{equation}\label{Beldec1}
\mathfrak{A}_{ac}=X_{abcd}t^{b}t^{d},\quad \mathfrak{B}_{ac} =-\overrightarrow{\star}X_{abcd} t^{b}t^{d},\quad\mathfrak{C}_{ac}= -\overleftarrow{\star} X_{abcd} t^{b}t^{d},\quad \mathfrak{D}_{ac} =\overleftrightarrow{\star} X_{abcd} t^{b}t^{d}.
\end{equation}
The latter are orthogonal to the vector $t^{q}$ by construction and each of them has a total of $9$ independent components, as expected. Simple manipulations of Eq. (\ref{Beldec}) then give the more compact relation
\begin{equation}\label{Beldec*}
X_{abcd}=\{\mathfrak{A}_{[a[c}t_{b]d]}+\overrightarrow{\star}\mathfrak{B}_{[a[c}t_{b]d]}+\overleftarrow{\star}\mathfrak{C}_{[a[c}t_{b]d]}+\overleftrightarrow{\star}\mathfrak{D}_{[a[c}t_{b]d]}\}/(t_{m}t^{m})^{2},
\end{equation}
where anti-symetrization is supposed to act solely on the pair $[ab]$ and $[cd]$ individually and $t_{ab}\equiv t_{a}t_{b}$, for conciseness (see Appendix C). In this last form we identify four independent double forms constructed with Kulkarni-Nomizu products involving the $2$-index tensors. It can be verified that the corresponding subspaces are orthogonal in the sense that
\begin{equation}\label{invariant1}
I\equiv\frac{1}{4}X_{abcd}X^{abcd}=(\mathfrak{A}_{ac}\mathfrak{A}^{ac}-\mathfrak{B}_{ac}\mathfrak{B}^{ac}-\mathfrak{C}_{ac}\mathfrak{C}^{ac}+\mathfrak{D}_{ac}\mathfrak{D}^{ac})/(t_{m}t^{m})^{4}.
\end{equation}
In other words, only the squared norms of the individual double forms contribute to the total result. At this point, it is important to emphasize that Eqs. (\ref{Beldec}) and (\ref{Beldec*}) generalise the classical Bel-Matte decomposition in two independent ways:
\begin{itemize}
\item{it is valid for an arbitrary $(2,2)$ double form;}
\item{it is valid for an arbitrary non-lightlike vector.}
\end{itemize}
In most of physical applications, the auxiliary quantity $t^{q}$ is identified with a time-like, future-oriented and normalized vector. In this situation, the projected tensors become spacelike, and can be seen as linear operators acting on three-dimensional vectors. Hence, they satisy the inequalities
\begin{equation}
\mathfrak{A}_{ac}\mathfrak{A}^{ac}\geq 0,\quad\quad\mathfrak{B}_{ac}\mathfrak{B}^{ac}\geq 0,\quad\quad\mathfrak{C}_{ac}\mathfrak{C}^{ac}\geq 0,\quad\quad\mathfrak{D}_{ac}\mathfrak{D}^{ac}\geq 0.
\end{equation}
Let us notice, however, that the set of $2$-index tensors $\{\mathfrak{A},\mathfrak{B},\mathfrak{C},\mathfrak{D}\}$ is still well defined for a lightlike vector. However, in this case, Eqs. (\ref{Beldec}) and (\ref{Beldec*}) are no longer valid, since they somehow degenerate. We refer the reader to \cite{Senovilla2} where a wealth of details concerning the decomposition of more general tensors in $n$-dimensional spacetimes of Lorentzian signature are discussed. 
\section{Quadratic identities}

In the early seventies, Lovelock \cite{Lovelock} had pointed out that a number of apparently unrelated tensor identities well known to relativists had a common underlying algebraic structure. This class of identities were called by him dimensionally dependent identities and had its roots in antisymmetrising over $n+1$ indices in an $n$-dimensional space. More recently, Brian Edgar and Hoglund \cite{Edgar1} have generalised Lovelock's results to all trace-free $(k,l)$ forms and have applied their master identity to Maxwell, Lanczos, Ricci, Bel and Bel-Robinson tensors (see also \cite{Edgar2}). As a byproduct, they also demonstrate how relationships between scalar invariants (syzygies) of the Riemann tensor can be investigated in a systematic manner. A simple example of such identities in the case of a $n\times n$ matrix $M^{a}_{\phantom a b}$ is the following relation
\begin{equation}
M^{c_{1}}_{\phantom a\phantom a [c_{1}}M^{c_{2}}_{\phantom a\phantom a c_{2}}...M^{c_{n}}_{\phantom a\phantom a c_{n}}\delta^{a}_{\phantom a b]}=0,
\end{equation}
which is nothing but the Cayley-Hamilton theorem in a disguised form. Other similar relations in the case of $(2,2)$ double forms in four dimensions are
\begin{equation}
X^{ab}_{\phantom a\phantom a [ab}X^{cd}_{\phantom a\phantom a cd}X^{e}_{\phantom a e]}=0,\quad\quad\quad X^{ab}_{\phantom a\phantom a [ab}X^{cd}_{\phantom a\phantom a cd}X^{ef}_{\phantom a\phantom a ef]}=0.
\end{equation}
In this section, we briefly explore some quadratic identities which are somehow related to the above ones. We shall see that such relations are all particular instances of the following proposition. 
\begin{prop}[Fundamental quadratic identity]
Let $A^{\{\Omega\}}_{\phantom a\phantom a\phantom a ab}$ and $B^{\{\Upsilon\}}_{\phantom a\phantom a\phantom a ab}$ denote two tensors, with $\{\Omega\}$ and $\{\Upsilon\}$ standing for an arbitrary number of additional lower and/or upper indices and a pair of skew indices $[ab]$. Defining the ``trace-free'' auxiliary tensor
\begin{equation}
\hat{C}^{\{\Omega,\Upsilon\}}_{\phantom a\phantom a\phantom a\phantom a\phantom a ab}=A^{\{\Omega\}}_{\phantom a\phantom a\phantom a ac}B^{\{\Upsilon\}c}_{\phantom a\phantom a\phantom a \phantom a b}+\frac{1}{4}\left(A^{\{\Omega\}rs}B^{\{\Upsilon\}}_{\phantom a\phantom a\phantom a rs}\right)g_{ab},\quad\quad\quad \hat{C}^{\{\Omega,\Upsilon\}r}_{\phantom a\phantom a\phantom a\phantom a\phantom a\phantom a r}=0,
\end{equation}
there follows the fundamental quadratic identity
\begin{equation}\label{qfi}
\left(A^{\{\Omega\}}_{\phantom a\phantom a\phantom a ab}\star\right)\left(B^{\{\Upsilon\}cd}\star\right)+A^{\{\Omega\}cd}B^{\{\Upsilon\}}_{\phantom a\phantom a\phantom a ab}=-\delta^{[c}_{\phantom a [a}\hat{C}^{\{\Omega,\Upsilon\}d]}_{\phantom a\phantom a\phantom a\phantom a\phantom a\phantom a\phantom a b]}.
\end{equation}
\end{prop}
\begin{proof}
The proof proceeds very much in the same way as in the case of Ruse-Lanczos identity. Simply write the duals above in terms of the generalized Kronecker delta and absorb it appropriately in the corresponding equation to obtain the remaining terms.
\end{proof}
\noindent We notice that the l.h.s. of Eq. (\ref{qfi}) contains no index summations. However, the roles of the skew pairs $[ab]$ and $[cd]$ are interchanged on the second term. Conversely, the r.h.s. of Eq. (\ref{qfi}) contains two types of contractions which emerge from the definition of the auxiliary tensor above: a first term involving a single index contraction and another involving two index contractions multiplied by the metric tensor. Furthermore, since the auxiliary tensor is trace-less, we have the following particular cases:

\begin{cor}
Contracting $b$ with $d$ in Eq. (\ref{qfi}) and rearranging the terms, there follows
\begin{equation}\label{emr}
\left(A^{\{\Omega\}}_{\phantom a\phantom a\phantom a ar}\star\right)\left(B^{\{\Upsilon\}cr}\star\right)-A^{\{\Omega\}cr}B^{\{\Upsilon\}}_{\phantom a\phantom a\phantom a ar}=-\frac{1}{2}\left(A^{\{\Omega\}rs}B^{\{\Upsilon\}}_{\phantom a\phantom a\phantom a rs}\right)\delta^{c}_{\phantom a a}. 
\end{equation}
\end{cor}

\begin{cor}
Substitution of $B^{\{\Upsilon\}}$ by $B^{\{\Upsilon\}}\star$ in the above equation gives
\begin{equation}\label{emr2}
\left(A^{\{\Omega\}}_{\phantom a\phantom a\phantom a ar}\star\right)B^{\{\Upsilon\}cr}+\left(B^{\{\Upsilon\}}_{\phantom a\phantom a\phantom a ar}\star\right)A^{\{\Omega\}cr}=+\frac{1}{2}\left(A^{\{\Omega\}rs}B^{\{\Upsilon\}}_{\phantom a\phantom a\phantom a rs}\star\right)\delta^{c}_{\phantom a a}. 
\end{equation}
\end{cor}
\noindent Quite often, the above relations may be used to simplify tensor expressions involving either pairs of $2$-forms, pairs of $(2,2)$ double forms or combinations of both. In what follows, we shall explore some consequences of these identities in simple particular situations.

\subsection{The case of $2$-forms}

In the study of the electromagnetic field one is concerned with the Faraday tensor $F_{ab}=-F_{ba}$ and its Hodge dual $\star F_{ab}=-\star F_{ba}$, also called the Maxwell tensor. From this pair of $2$-forms, one defines two independent algebraic invariants
\begin{equation}
\psi\equiv \frac{1}{2}F_{ab}F^{ab},\quad\quad\quad \phi\equiv \frac{1}{2}\star F_{ab}F^{cd}.
\end{equation}
Making the substitutions $A^{\{\Omega\}}_{\phantom a\phantom a\phantom a ab}\mapsto F_{ab}$ and $B^{\{\Omega\}}_{\phantom a\phantom a\phantom a ab}\mapsto F_{ab}$ in Proposition 1, one obtains the identity
\begin{equation}\label{emid1}
(\star F_{ab})(\star F^{cd})+F_{ab}F^{cd}=-\delta^{[c}_{\phantom a [a}T^{d]}_{\phantom a\phantom a b]}
\end{equation}
where the trace-less symmetric quantity
\begin{equation}
T_{ab}=F_{ar}F^{r}_{\phantom a b}+\frac{1}{2}\psi g_{ab}
\end{equation}
is nothing but the energy-momentum tensor of the electromagnetic field. We notice that the r.h.s. of Eq. (\ref{emid1}) is itself an element of the symmetric  subspace $^{(s)}\mbox{KN1}$, as defined in Section II. Indeed, it is symmetric under the interchange $ab \leftrightarrow cd$ and contains only the first trace, since the energy-momentum tensor is trace-less. Conversely, If we Hodge dualize any pair of skew indices above, we get
\begin{equation}\label{emid2}
\star F_{ab}F^{cd}-\star F^{cd}F_{ab}=\varepsilon_{ab}^{\phantom a\phantom a q[c}T^{d]}_{\phantom a q}, 
\end{equation}
which defines an anti-symmetric double form belonging to the subspace $^{(a)}\mbox{Weyl}$. Curiously enough, one hardly sees Eqs. (\ref{emid1}) and (\ref{emid2}) in textbooks of relativistic electrodynamics. The reason behind this is, we believe, that their derivations require some manipulations with the multi-indices described in Proposition V.1. An interesting feature of the above equations is that they are invariant under the dual rotation, defined by the map of the space of 2-forms onto itself
\begin{equation}\label{dualityr}
F_{ab}\quad\mapsto\quad\mbox{cos}(\theta)F_{ab}+\mbox{sin}(\theta)\star F_{ab},\quad\quad\quad \star F_{ab}\quad\mapsto\quad-\mbox{sin}(\theta)F_{ab}+\mbox{cos}(\theta)\star F_{ab},
\end{equation} 
where $\theta\in\mathbb{R}$. This fact becomes more transparent if we take the trace of Eq. (\ref{emid1}), to write the energy-momentum tensor as
\begin{equation}\label{emta}
T_{ab}=\frac{1}{2}\big(\star F_{ar}\star F^{r}_{\phantom a b}+F_{ar}F^{r}_{\phantom a b}\big).
\end{equation}
Indeed, direct substitution of Eq. (\ref{dualityr}) into Eq. (\ref{emta}) gives us the invariance of Eqs. (\ref{emid1}) and (\ref{emid2}) under dual rotations. It is remarkable that the energy-momentum tensor of the electromagnetic field emerges as a by-product of the algebraic identity Eq. (\ref{qfi}). 

Let us now move on to Corolaries V.1.1 and V.1.2. Putting $A^{\{\Omega\}}_{\phantom a\phantom a\phantom a ab}\mapsto F_{ab}$ and $B^{\{\Omega\}}_{\phantom a\phantom a\phantom a ab}\mapsto F_{ab}$ in Eqs. (\ref{emr}) and (\ref{emr2}), one gets the quadratic relations
\begin{equation}
(\star F^{ab})(\star F_{bc})- F^{ab} F_{bc}=\psi \delta^{a}_{\phantom a c}, 
\end{equation}
\begin{equation}
\star F^{ab}F_{bc}=-\frac{1}{2}\phi\delta^{a}_{\phantom a c}. 
\end{equation}
From the latter, it is possible to show also the third and quartic identities
\begin{eqnarray}
F^{a}_{\phantom a p}F^{p}_{\phantom a q}F^{q}_{\phantom a b}+\psi F^{a}_{\phantom a b}+\frac{1}{2}\phi\star F^{a}_{\phantom a b} &=&0,\\
F^{a}_{\phantom a p}F^{p}_{\phantom a q}F^{q}_{\phantom a r}F^{r}_{\phantom a b}+\psi F^{a}_{\phantom a p}F^{p}_{\phantom a b}-\frac{1}{4}\phi^{2}\delta^{a}_{\phantom a b}&=&0.
\end{eqnarray}
We notice that the last equation is nothing but a consequence of the Cayley-Hamilton theorem. In particular, it implies that $\mbox{det}(F^{a}_{\phantom a b})=-\phi^{2}/4$. The above relations have proved to be useful in the study of the algebraic features of the Faraday tensor. Let us now investigate similar relations satisfied by arbitrary $(2,2)$ double forms.



\subsection{The case of $(2,2)$ double forms}

To begin with, we let $X_{a_{1}a_{2}b_{1}b_{2}}$ and $Y_{c_{1}c_{2}d_{1}d_{2}}$ denote two arbitrary elements of $\mathcal{D}^{2,2}$. Making the appropriate substitutions in Proposition V.1, we obtain
\begin{equation}\label{quadridentity}
\left(\overrightarrow{\star}X^{a_{1}a_{2}}_{\phantom a\phantom a\phantom a\phantom a b_{1}b_{2}}\right)\left(\overrightarrow{\star}Y^{c_{1}c_{2}}_{\phantom a\phantom a\phantom a\phantom a d_{1}d_{2}}\right)+X^{a_{1}a_{2}}_{\phantom a\phantom a\phantom a\phantom a d_{1}d_{2}}Y^{c_{1}c_{2}}_{\phantom a\phantom a\phantom a\phantom a b_{1}b_{2}}=-g_{[d_{1}[b_{1}}\hat{C}^{a_{1}a_{2}c_{1}c_{2}}_{\phantom a\phantom a\phantom a\phantom a\phantom a\phantom a \phantom a\phantom a d_{2}]b_{2}]},
\end{equation}
where
\begin{equation}
\hat{C}^{a_{1}a_{2}c_{1}c_{2}}_{\phantom a\phantom a\phantom a\phantom a\phantom a\phantom a\phantom a\phantom a d_{2}b_{2}}=-X^{a_{1}a_{2}k_{1}}_{\phantom a\phantom a\phantom a\phantom a\phantom a\phantom a d_{2}}Y^{c_{1}c_{2}}_{\phantom a\phantom a\phantom a \phantom a k_{1}b_{2}}+\frac{1}{4}\left(X^{a_{1}a_{2}k_{1}k_{2}}Y^{c_{1}c_{2}}_{\phantom a\phantom a\phantom a\phantom a k_{1}k_{2}}\right)g_{d_{2}b_{2}}
\end{equation}
and (again) anti-symmetrization is supposed to act solely on the pairs
$[b_{1}b_{2}]$ and $[d_{1}d_{2}]$ individually. It is clear that the above relation is concerned with elements of the space $\mathcal{D}^{2,2}\otimes\mathcal{D}^{2,2}$ and that several new possible relations may be obtained by taking either combinations of Hodge duals or particular index permutations. 

To the best of our knowledge, Eq. (\ref{quadridentity}) does not appear in this form in previous literature. However, we shall see that it is at the roots of some known identities, such as those discussed by Lanczos \cite{Lanczos1}, Debever \cite{Debever1} and Novello-Duarte \cite{Novello}. In spite of the fact that these identities were originally obtained in the case of the Weyl tensor, we shall see that they remain valid for any \textit{irreducible} $(2,2)$ double form. In order to prove this, we start by restricting Eq. (\ref{quadridentity}) to the case of a single double form $X_{abcd}$. Contracting $a_{2}$ with $c_{2}$ and $b_{2}$ with $d_{2}$, we get
\begin{equation}\label{quad1}
\left(\overrightarrow{\star}X^{a_{1}\phantom a\phantom a b_{1}}_{\phantom a\phantom a k_{1}\phantom a\phantom a k_{2}}\right)\left(\overrightarrow{\star}X_{c_{1}\phantom a\phantom a d_{1}}^{\phantom a\phantom a k_{1}\phantom a\phantom a k_{2}}\right)-X^{a_{1}}_{\phantom a\phantom a k_{1}d_{1}k_{2}}X_{c_{1}}^{\phantom a k_{1}b_{1}k_{2}}=-\frac{1}{2}\left(X^{a_{1}k_{1}k_{2}k_{3}}X_{ c_{1}k_{1}k_{2}k_{3}}\right)\delta^{b_{1}}_{\phantom a d_{1}}.
\end{equation}
Further contractions then results in the following relations
\begin{equation}\label{trivial1}
\left(\overrightarrow{\star}X^{a_{1}}_{\phantom a\phantom a k_{1} k_{2}k_{3}}\right)\left(\overrightarrow{\star}X_{c_{1}}^{\phantom a\phantom a k_{1}k_{2}k_{3}}\right)=-X^{a_{1}}_{\phantom a\phantom a k_{1} k_{2}k_{3}}X_{c_{1}}^{\phantom a\phantom a k_{1}k_{2}k_{3}},\\
\end{equation}
\begin{equation}\label{quad2}
\left(\overrightarrow{\star}X_{k_{1}k_{2} \phantom a\phantom a k_{3}}^{\phantom a\phantom a\phantom a\phantom a b_{1}}\right)\left(\overrightarrow{\star}X^{k_{1}k_{2}\phantom a\phantom a k_{3}}_{\phantom a\phantom a\phantom a\phantom a d_{1}}\right)-X_{k_{1}k_{2}d_{1}k_{3}}X^{k_{1}k_{2}b_{1}k_{3}}=-2I\delta^{b_{1}}_{\phantom a d_{1}}.
\end{equation}
with the invariant defined as in Eq. (\ref{invariant1}). We now restrict ourselves to the particular case of irreducible $(2,2)$ double forms i.e., we assume that $X_{abcd}$ belongs to one of the subspaces as presented in Eq. (\ref{irreducfour}). We then get the following results

\begin{prop}[Generalised Lanczos identity]
Every irreducible (2,2) double form $X_{abcd}$ satisfies the relation
\begin{equation}
X^{a_{1}k_{1}k_{2}k_{3}}X_{b_{1}k_{1}k_{2}k_{3}}=I\delta^{a_{1}}_{\phantom a b_{1}},
\end{equation}

\end{prop}
\begin{proof}
An irreducible double form is either symmetric or antisymmetric, and the same holds for every of its possible Hodge duals. Using this property to interchange the skew pairs in Eq. (\ref{quad2}) and combining the latter with Eq. (\ref{trivial1}) then gives the desired result. 
\end{proof}

\begin{prop}[Generalised Debever identity]
Every irreducible (2,2) double form $X_{abcd}$ satisfies the relation
\begin{equation}
\left(\overrightarrow{\star}X^{a_{1}\phantom a\phantom a b_{1}}_{\phantom a\phantom a k_{1}\phantom a\phantom a k_{2}}\right)\left(\overrightarrow{\star}X_{c_{1}\phantom a\phantom a d_{1}}^{\phantom a\phantom a k_{1}\phantom a\phantom a k_{2}}\right)-X^{a_{1}}_{\phantom a\phantom a k_{1}d_{1}k_{2}}X_{c_{1}}^{\phantom a k_{1}b_{1}k_{2}}=-\frac{1}{2}I\delta^{a_{1}}_{\phantom a c_{1}}\delta^{b_{1}}_{\phantom a d_{1}}.
\end{equation}

\end{prop}
\begin{proof}
Simply substitute the generalised Lanczos identity above into Eq. (\ref{quad1}) to get the desired result.
\end{proof}

 Finally, to generalise the quadratic identities presented by Novello and Duarte in \cite{Novello}, we first use Corollary V.1.2 and perform another contraction to write
\begin{equation}
\left(\overrightarrow{\star}X^{a_{1}\phantom a\phantom a b_{1}}_{\phantom a\phantom a k_{1}\phantom a\phantom a k_{2}}\right)\left(X_{c_{1}\phantom a\phantom a d_{1}}^{\phantom a\phantom a k_{1}\phantom a\phantom a k_{2}}\right)+X^{a_{1}}_{\phantom a\phantom a k_{1}d_{1}k_{2}}\overrightarrow{\star}X_{c_{1}}^{\phantom a k_{1}b_{1}k_{2}}=\frac{1}{2}\left(X^{a_{1}k_{1}k_{2}k_{3}}\overrightarrow{\star}X_{ c_{1}k_{1}k_{2}k_{3}}\right)\delta^{b_{1}}_{\phantom a d_{1}}.
\end{equation}
Simple manipulations then give, for an irreducible (2,2) form, the ``mixed'' relations
\begin{equation}
\overrightarrow{\star}X^{a_{1}k_{1}k_{2}k_{3}}X_{b_{1}k_{1}k_{2}k_{3}}=J\delta^{a_{1}}_{\phantom a b_{1}},
\end{equation}
\begin{equation}
\left(\overrightarrow{\star}X^{a_{1}\phantom a\phantom a b_{1}}_{\phantom a\phantom a k_{1}\phantom a\phantom a k_{2}}\right)\left(X_{c_{1}\phantom a\phantom a d_{1}}^{\phantom a\phantom a k_{1}\phantom a\phantom a k_{2}}\right)+X^{a_{1}}_{\phantom a\phantom a k_{1}d_{1}k_{2}}\overrightarrow{\star}X_{c_{1}}^{\phantom a k_{1}b_{1}k_{2}}=2J\delta^{a_{1}}_{\phantom a c_{1}}\delta^{b_{1}}_{\phantom a d_{1}},
\end{equation}
with
\begin{equation}
J\equiv \frac{1}{4}X^{k_{1}k_{2}k_{3}k_{4}}\overrightarrow{\star}X_{ k_{1}k_{2}k_{3}k_{4}},
\end{equation}
which completes our claim. It is worth emphasizing that Novello and Duarte obtain these identities taking full advantage of the traceless property of the Weyl tensor. However, our approach shows that this assumption is not necessary.

\section{Conclusions}

The net result of this paper is the derivation/discussion of several algebraic relations identically satisfied by generic $(2,2)$ double forms. Although some of these identities are known to hold in the particular cases of Weyl and/or Riemann tensors, their extensions to the general situation do not appear in a self-contained way in previous literature. In order to make our discussion as general as possible, we have started with a mathematical setting in arbitrary dimensions
and have introduced some non-standard terminology. We then discussed the internal symmetries of these tensors, characterised the properties of their invariant sub-spaces and have emphasized the naturalness of Kulkarni-Nomizu products in this context. A key ingredient in our subsequent discussion is the fact that the space of (2,2) double forms splits into an orthogonal direct sum of six invariant subspaces. 

Restricting ourselves to the case of a four-dimensional spacetime, we have introduced the concept of Hodge duality operators and carefully analysed the action of the corresponding maps on the invariant subspaces. An important consequence of this analysis is the fact that Hodge star operators map irreducible pieces into irreducible pieces. As stressed in the text, this result is a direct consequence of the generalised version of the Ruse-Lanczos identity which permits a direct characterization of the corresponding target spaces. We then generalised other important relations such as the the Bel-Matte decomposition of a $(2,2)$ double form into space-like pieces and the Lovelock-like quadratic identities. In particular, we hope that the latter may be of crucial importance in the studies of i) covariant hyperbolizations of electrodynamics inside media; ii) alternative derivations of the Tamm-Rubilar tensor using traditional methods; iii) algebraic classification of local constitutive tensors. We shall investigate some of these issues in a forthcoming communication.


\appendix
\section{Generalized Kronecker deltas and Levi-Civita tensors}\label{App1}

The generalized Kronecker delta of order $2k$ is defined by the $k\times k$ multi-linear determinant
\begin{equation}
\delta^{a_{1}...a_{k}}_{\phantom a\phantom a\phantom a\phantom a \phantom a \phantom a b_{1}...b_{k}}\equiv \mbox{det}\left(\begin{array}{cccc}
\delta^{a_{1}}_{\phantom a\phantom a b_{1}}& \delta^{a_{1}}_{\phantom a\phantom a b_{2}}&\cdots &\delta^{a_{1}}_{\phantom a\phantom a b_{k}}\\
\delta^{a_{2}}_{\phantom a\phantom a b_{1}}&\delta^{a_{2}}_{\phantom a\phantom a b_{2}}&\cdots &\delta^{a_{2}}_{\phantom a\phantom a b_{k}}\\
\vdots &\vdots &\ddots &\vdots \\
\delta^{a_{k}}_{\phantom a\phantom a b_{1}}&\delta^{a_{k}}_{\phantom a\phantom a b_{2}}&\cdots &\delta^{a_{k}}_{\phantom a\phantom a b_{k}}\\
\end{array}\right).
\end{equation}
Using the Laplace expansion of determinant, it may be defined recursively and the pattern of the first few terms read as follows
\begin{eqnarray*}
&& \delta^{a_{1}a_{2}}_{\phantom a\phantom a\phantom a\phantom a b_{1}b_{2}}=\delta^{a_{1}}_{\phantom a\phantom a b_{1}}\delta^{a_{2}}_{\phantom a\phantom a b_{2}}-\delta^{a_{1}}_{\phantom a\phantom a b_{2}}\delta^{a_{2}}_{\phantom a\phantom a b_{1}},\\
&& \delta^{a_{1}a_{2}a_{3}}_{\phantom a\phantom a\phantom a\phantom a\phantom a\phantom a b_{1}b_{2}b_{3}}=\delta^{a_{1}}_{\phantom a\phantom a b_{1}}\delta^{a_{2}a_{3}}_{\phantom a\phantom a\phantom a\phantom a b_{2}b_{3}}+\delta^{a_{1}}_{\phantom a\phantom a b_{2}}\delta^{a_{2}a_{3}}_{\phantom a\phantom a\phantom a\phantom a b_{3}b_{1}}+\delta^{a_{1}}_{\phantom a\phantom a b_{3}}\delta^{a_{2}a_{3}}_{\phantom a\phantom a\phantom a\phantom a b_{1}b_{2}},\\
&& \delta^{a_{1}a_{2}a_{3}a_{4}}_{\phantom a\phantom a\phantom a\phantom a\phantom a\phantom a\phantom a\phantom a b_{1}b_{2}b_{3}b_{4}}=\delta^{a_{1}}_{\phantom a\phantom a b_{1}}\delta^{a_{2}a_{3}a_{4}}_{\phantom a\phantom a\phantom a\phantom a\phantom a\phantom a b_{2}b_{3}b_{4}}-\delta^{a_{1}}_{\phantom a\phantom a b_{2}}\delta^{a_{2}a_{3}a_{4}}_{\phantom a\phantom a\phantom a\phantom a\phantom a\phantom a b_{3}b_{4}b_{1}}+\delta^{a_{1}}_{\phantom a\phantom a b_{3}}\delta^{a_{2}a_{3}a_{4}}_{\phantom a\phantom a\phantom a\phantom a\phantom a\phantom a b_{4}b_{1}b_{2}}-\delta^{a_{1}}_{\phantom a\phantom a b_{4}}\delta^{a_{2}a_{3}a_{4}}_{\phantom a\phantom a\phantom a\phantom a\phantom a\phantom a b_{1}b_{2}b_{3}}.
\end{eqnarray*}
Given a covariant tensor of the type $Y_{a_{1}...a_{k}}$, total anti-symmetrization is defined by means of square brackets as
\begin{equation}
Y_{[a_{1}...a_{k}]}\equiv \delta^{b_{1}...b_{k}}_{\phantom a\phantom a\phantom a\phantom a \phantom a \phantom a a_{1}...a_{k}}Y_{b_{1}...b_{k}}.
\end{equation}
Throughout, the four-dimensional Levi-Civita tensors are given by
\begin{equation}
\varepsilon_{abcd}\equiv\sqrt{-g}[abcd],\quad\quad\quad \varepsilon^{abcd}\equiv-\frac{1}{\sqrt{-g}}[abcd],
\end{equation}
with $g\equiv\mbox{det}(g_{ab})$ and $[abcd]$ the totally antisymmetric symbol, with $[0123]=+1$. The fundamental relation between these tensors reads as
\begin{equation}
\varepsilon^{a_{1}a_{2}a_{3}a_{4}}\varepsilon_{b_{1}b_{2}b_{3}b_{4}}=-\delta^{a_{1}a_{2}a_{3}a_{4}}_{\phantom a\phantom a\phantom a\phantom a\phantom a\phantom a\phantom a\phantom a b_{1}b_{2}b_{3}b_{4}}.
\end{equation}

\section{Kulkarni-Nomizu products}
Let $T_{2}$ denote the space of covariant rank-2 tensors at a spacetime point. If $\boldsymbol{k}, \boldsymbol{h}\in T_{2}$, the Kulkarni-Nomizu product of $\boldsymbol{k}$ and $\boldsymbol{h}$ is a symmetric map $\KN:T_{2}\times T_{2}\rightarrow \mathcal{D}^{2,2}$, defined by\\
\begin{eqnarray*}
(\boldsymbol{k} \KN \boldsymbol{h})_{abcd}\equiv k_{[a[c}h_{b]d]}=k_{ac}h_{bd}-k_{ad}h_{bc}+k_{bd}h_{ac}-k_{bc}h_{ad}.
\end{eqnarray*}\\
A direct consequence of the definition is that
\begin{equation}
(\boldsymbol{k} \KN \boldsymbol{h})_{abcd}=(\boldsymbol{h} \KN \boldsymbol{k})_{abcd},\quad\quad\quad(\boldsymbol{k} \KN \boldsymbol{k})_{abcd}=2k_{a[c}k_{bd]},
\end{equation}
\begin{equation}
(\boldsymbol{k} \KN \boldsymbol{h})_{a[pqr]}=-2(k_{ap}h_{[qr]}+k_{aq}h_{[rp]}+k_{ar}h_{[pq]}+h_{ap}k_{[qr]}+h_{aq}k_{[rp]}+h_{ar}k_{[pq]}).
\end{equation}
More specifically, if we consider the decomposition $T_{2}=\mbox{Sym}(T_{2})\oplus \mbox{Skew}(T_{2})$, we have the following possibilities for the corresponding images\\

\begin{center}
  \begin{tabular}{ | l || c | r | }
    \hline 
    \diagbox[innerleftsep=.7cm, innerrightsep=.5cm,width=6em]{$\boldsymbol{k}$}{$\boldsymbol{h}$} & $\mbox{Sym}(T_2)$ & $\mbox{Skew}(T_2)\ \ $\\
\hline\hline
$\ \ \ \mbox{Sym}(T_{2})$ & $\ \mbox{Curv}(\mathfrak {D}^{2,2})\ $ & $\ \mbox{Skew}(\mathfrak {D}^{2,2})\ $ \\
\hline
$\ \ \ \mbox{Skew}(T_2)$ & $\ \mbox{Skew}(\mathfrak {D}^{2,2})\ $ & $\ \mbox{Sym}(\mathfrak {D}^{2,2})\ $ \\
    \hline
  \end{tabular}
\end{center}

As usual, the bi-metric tensor is denoted by 
\begin{equation}
g_{abcd}\equiv g_{a[c}g_{bd]}=g_{ac}g_{bd}-g_{ad}g_{cd},\quad\quad\quad g_{abcd}=g_{cdab},
\end{equation}
and some relations involving contractions of the latter read as
\begin{eqnarray}
&&g^{bd}(g_{abpq}g_{cdrs})=g_{aspq}g_{cr}-g_{arpq}g_{cs},\\
&&g^{bd}(g_{abpq}\varepsilon_{cdrs})=\varepsilon_{cqrs}g_{ap}-\varepsilon_{cprs}g_{aq},\\
&&g^{bd}(\varepsilon_{abpq}g_{cdrs})=\varepsilon_{aspq}g_{cr}-\varepsilon_{arpq}g_{cs}.
\end{eqnarray}

\section{Matrix displays}

If we define collective indices $\alpha,\beta$ taking values in the
set
\begin{equation}
1\rightarrow [01]\quad 2\rightarrow [02]\quad3\rightarrow [03]\quad 4\rightarrow [32]\quad 5\rightarrow [13]\quad 6\rightarrow [21],
\end{equation}
Eqs. (\ref{Beldec}) and (\ref{Beldec*}) permit us to represent the $(2,2)$ double form by a $6\times 6$ block matrix called the \textit{Petrov equivalent} of $X_{abcd}$. This is generally done by choosing an orthonormal frame with a convenient orientation, such that the timelike vector is identified with the first leg of the tetrad. A direct calculation then gives 
\begin{equation}
 X_{abcd}\quad\leftrightarrow\quad X_{\alpha\beta}= 
  \begin{pmatrix}
    \mathfrak{A} & \mathfrak{B}\\
    \mathfrak{C} & \mathfrak{D}
  \end{pmatrix},\quad\quad\quad X_{cdab}\quad\leftrightarrow\quad X_{\beta\alpha}= 
  \begin{pmatrix}
    \mathfrak{A}^{T} & \mathfrak{C}^{T}\\
    \mathfrak{B}^{T} & \mathfrak{D}^{T}
  \end{pmatrix},
\end{equation}
with the quantities $\{\mathfrak{A},\mathfrak{B},\mathfrak{C},\mathfrak{D}\}$ denoting the covariant $3\times 3$ matrices constructed with the projected tensors Eqs. (\ref{Beldec1}) in the obvious way. In particular, for the following Petrov equivalents
\begin{equation}
g_{abcd}\quad\leftrightarrow\quad g_{\alpha\beta},\quad g^{abcd}\quad\leftrightarrow\quad g^{\alpha\beta},\quad\delta^{ab}_{\phantom a\phantom a cd}\quad\leftrightarrow\quad\delta^{\alpha}_{\phantom a\beta},\quad \delta_{ab}^{\phantom a\phantom a cd}\quad\leftrightarrow\quad\delta_{\alpha}^{\phantom a\beta},
\end{equation}
we obtain
\begin{equation}
g_{\alpha\beta}= 
  \begin{pmatrix}
    -\boldsymbol{1} & \ \ \boldsymbol{0}\\
    \ \ \boldsymbol{0} & +\boldsymbol{1}
  \end{pmatrix},\quad\quad g^{\alpha\beta}= 
  \begin{pmatrix}
    -\boldsymbol{1} & \ \ \boldsymbol{0}\\
    \ \ \boldsymbol{0} & +\boldsymbol{1}
  \end{pmatrix},\quad\quad  \delta^{\alpha}_{\phantom a\beta}=\begin{pmatrix}
    +\boldsymbol{1} & \ \ \boldsymbol{0}\\
    \ \ \boldsymbol{0} & +\boldsymbol{1}
  \end{pmatrix},\quad\quad\quad \delta_{\alpha}^{\phantom a\beta}=\begin{pmatrix}
    +\boldsymbol{1} & \ \ \boldsymbol{0}\\
    \ \ \boldsymbol{0} & +\boldsymbol{1}
  \end{pmatrix},
\end{equation}
with $\boldsymbol{1}$ and $\boldsymbol{0}$ denoting the corresponding $3\times 3$ identity and zero matrices. Since $g_{\alpha\beta}$ plays the role of a metric in the space of bi-vectors we call it simply by bi-metric\footnote{A six-dimensional vector space with metric signature $(3,3)$ is often referred in literature as the Klein space. Consequently, double (2,2) forms may be thought as rank-2 covariant tensors in this space \cite{Pirani}.}. Similarly, for the Levi-Civita tensors, there follow
\begin{equation}
\varepsilon_{\alpha\beta}= 
  \begin{pmatrix}
    \ \ \boldsymbol{0} & -\boldsymbol{1}\\
    -\boldsymbol{1} & \ \ \boldsymbol{0}
  \end{pmatrix},\quad\quad\quad\varepsilon^{\alpha}_{\phantom a\beta}= 
  \begin{pmatrix}
    \ \ \boldsymbol{0} & +\boldsymbol{1}\\
    -\boldsymbol{1} & \ \ \boldsymbol{0}
  \end{pmatrix},\quad\quad\quad \varepsilon_{\alpha}^{\phantom a\beta}= 
  \begin{pmatrix}
    \ \ \boldsymbol{0} & -\boldsymbol{1}\\
    +\boldsymbol{1} & \ \ \boldsymbol{0}
  \end{pmatrix},\quad\quad\quad
\varepsilon^{\alpha\beta}=
  \begin{pmatrix}
    \ \ \boldsymbol{0} & +\boldsymbol{1}\\
    +\boldsymbol{1} & \ \ \boldsymbol{0}
  \end{pmatrix},\end{equation}
 and the fundamental relation concerning the above quantities reads as
\begin{equation}
  \frac{1}{2}g^{abpq}g_{pqcd}=-\frac{1}{2}\varepsilon^{abpq}\varepsilon_{pqcd}=\delta^{ab}_{\phantom a\phantom acd}\quad\leftrightarrow\quad  g^{\alpha\gamma}g_{\gamma\beta}=-\varepsilon^{\alpha\gamma}\varepsilon_{\gamma\beta}=\delta^{\alpha}_{\phantom a\beta}.
\end{equation}

An immediate advantage of using Petrov equivalents is that initially complicated tensorial manipulations become simple matrix relations. In particular, the left and right Hodge duals read as
\begin{equation}
\star X_{abcd}=\frac{1}{2}\varepsilon_{ab}^{\phantom a\phantom a pq}X_{pqcd}\quad\leftrightarrow\quad \star X_{\alpha\beta}= \varepsilon_{\alpha}^{\phantom a\lambda}X_{\lambda\beta}=
  \begin{pmatrix}
    -\mathfrak{C} & -\mathfrak{D}\\
    \ \ \mathfrak{A} & \ \  \mathfrak{B}\end{pmatrix},   
\end{equation}
\begin{equation}
 X_{abcd}\star=\frac{1}{2}\varepsilon^{pq}_{\phantom a\phantom a cd}X_{abpq}\quad\leftrightarrow\quad X_{\alpha\beta}\star= X_{\alpha\lambda}\varepsilon^{\lambda}_{\phantom a\beta}=
   \begin{pmatrix}
    -\mathfrak{B} & \mathfrak{A}\\
    -\mathfrak{D} & \mathfrak{C}
  \end{pmatrix},
\end{equation}
 from which it become clear that left-dualization is a transformation on the rows and right-dualization is a transformation on the columns. Combining the above transformations, one easily obtains the double dual in the form
 \begin{equation}
\star X_{abcd}\star\quad\leftrightarrow\quad \star X_{\alpha\beta}\star= 
  \begin{pmatrix}
    \ \ \mathfrak{D} & -\mathfrak{C}\\
    -\mathfrak{B} & \ \ \mathfrak{A}
  \end{pmatrix}.
\end{equation}
Also, by noticing that the inverse of $\varepsilon^{\alpha}_{\phantom a\beta}$ is minus itself, we write 
\begin{equation}
\star X^{\alpha}_{\phantom a\beta}=-\varepsilon^{\alpha}_{\phantom a\mu}(X^{\mu}_{\phantom a\nu}\star)\varepsilon^{\nu}_{\phantom a\beta},\quad\quad\quad \star X^{\alpha}_{\phantom a\beta}\star=-\varepsilon^{\alpha}_{\phantom a\mu}(-X^{\mu}_{\phantom a\nu})\varepsilon^{\nu}_{\phantom a\beta},
\end{equation}
which shows that the corresponding tensors are related by similarity transformations. We shall see next that the above relations somehow simplify the analysis of the independent invariants which may be constructed with the double form and its duals.

Another nice feature of the present formalism manifests itself in the description of the irreducible decomposition provided by Eq. (\ref{Irreddec}). From the latter, we have the Petrov equivalents
\begin{equation}
X_{\alpha\beta}\ =\ ^{(1)}X_{\alpha\beta}\ +\ ^{(2)}X_{\alpha\beta}\ +\ ^{(3)}X_{\alpha\beta}\ +\ ^{(4)}X_{\alpha\beta}\ +\ ^{(5)}X_{\alpha\beta}\ +\ ^{(6)}X_{\alpha\beta}\,     
\end{equation}
and, since each subspace is characterized by four $3\times 3$ matrices, we are lead to investigate a total of six independent sets $\{\ ^{(i)}\mathfrak{A},\ ^{(i)}\mathfrak{B}\,\ ^{(i)}\mathfrak{C},\ ^{(i)}\mathfrak{D}\}$ with $i=1,2,...,6$. When combined with relations Eqs. (\ref{RL1})-(\ref{RL2}) and Eqs. (\ref{RL3})-(\ref{RL4}) for the Hodge duals, we realise that the elements of each set have specific algebraic properties. In particular, for the curvaturelike part, we obtain
\begin{equation}
^{(1)}X_{\alpha\beta}= \begin{pmatrix}
    ^{(1)}\mathfrak{A} & \ \ \ ^{(1)}\mathfrak{B}\\
    ^{(1)}\mathfrak{B} & - ^{(1)}\mathfrak{A}
  \end{pmatrix},\quad\quad ^{(2)}X_{\alpha\beta}= \begin{pmatrix}
    \ \ ^{(2)}\mathfrak{A} & \ \ ^{(2)}\mathfrak{B}\\
    -^{(2)}\mathfrak{B} &  \ \ ^{(2)}\mathfrak{A} 
  \end{pmatrix},\quad\quad
  ^{(3)}X_{\alpha\beta}=\frac{X}{12} \begin{pmatrix}
    -\boldsymbol{1} & \ \ \boldsymbol{0}\\
    \ \ \boldsymbol{0} & +\boldsymbol{1}
  \end{pmatrix},
\end{equation}
with $^{(1)}\mathfrak{A}$ and $^{(1)}\mathfrak{B}$ symmetric traceless, $^{(2)}\mathfrak{A}$ symmetric and $^{(2)}\mathfrak{B}$ antisymmetric totalizing with the second trace $20$ independent components. The four-form term is naturally proportional to the Levi-Civita tensor and is easily obtained as
\begin{equation}
^{(4)}X_{\alpha\beta}= \frac{\star X}{12}\begin{pmatrix}
    \ \ \boldsymbol{0} & +\boldsymbol{1}\\
    +\boldsymbol{1} & \ \ \boldsymbol{0}
  \end{pmatrix}.
\end{equation}
From the latter one sees that it is completely determined by the second trace of either the left or right Hodge dual. Finally, for the antisymmetric part, there follow
\begin{equation}
^{(5)}X_{\alpha\beta}= \begin{pmatrix}
    \ \ ^{(5)}\mathfrak{A} & \ \ \ ^{(5)}\mathfrak{B}\\
    -^{(5)}\mathfrak{B} &  \ \ ^{(5)}\mathfrak{A}
  \end{pmatrix},\quad\quad\quad ^{(6)}X_{\alpha\beta}= \begin{pmatrix}
    ^{(6)}\mathfrak{A} & \ \ \ ^{(6)}\mathfrak{B}\\
    ^{(6)}\mathfrak{B} &  \ \ ^{(6)}\mathfrak{A}
  \end{pmatrix},
\end{equation}
with $^{(5)}\mathfrak{A}$ antisymmetric, $^{(5)}\mathfrak{B}$ symmetric, $^{(6)}\mathfrak{A}$ and $^{(6)}\mathfrak{B}$ antisymmetric totalizing the remaining $15$ components.

\section{Main invariants}

The characteristic polynomial associated to $X^{\alpha}_{\phantom a\beta}$ reads as
\begin{equation}
\chi(\lambda)=\mbox{det}(X^{\alpha}_{\phantom a\beta}-\lambda\delta^{\alpha}_{\phantom a\beta})=\sum\limits_{k=0}^{6}\sigma_{k}(-\lambda)^{6-k}, \end{equation}
where $\sigma_{k}$ is the $k$-th elementary symmetric polynomial, given by
\begin{equation}
    \sigma_{k}\equiv\frac{1}{k!}\delta^{\alpha_{1}...\alpha_{6}}_{\phantom a\phantom a\phantom a\phantom a\phantom a\phantom a\beta_{1}...\beta_{6}}X^{\beta_{1}}_{\phantom a\phantom a\alpha_{1}}...X^{\beta_{k}}_{\phantom a\phantom a\alpha_{k}},
\end{equation}
with $\sigma_{0}\equiv1$. The algebraic complexity of the double form is directly related to the number of non-vanishing such polynomials and a tedious, but direct, calculation gives
\begin{eqnarray}
&&\sigma_{1}=\frac{1}{1!}[X],\\
&&\sigma_{2}=\frac{1}{2!}([X]^{2}-[X^{2}]),\\
&&\sigma_{3}=\frac{1}{3!}([X]^{3}-3[X][X^{2}]+2[X^{3}]),\\
&&\sigma_{4}=\frac{1}{4!}([X]^{4}-6[X]^{2}[X^{2}]+3[X^{2}]^{2}+8[X][X^{3}]-6[X^{4}]),\\
&&\sigma_{5}=\frac{1}{5!}([X]^{5}-10[X]^{3}[X^{2}]+15[X][X^{2}]^{2}+20[X]^{2}[X^{3}]-20[X^{2}][X^{3}]\\\nonumber
&&\quad\quad\quad\quad\quad\quad\quad\quad\quad\quad\quad\quad\quad-30[X][X^{4}]+24[X^{5}]),\\
&&\sigma_{6}=\frac{1}{6!}([X]^{6}-15[X]^{4}[X^{2}]+40[X]^{3}[X^{3}]+45[X]^{2}[X^{2}]^{2}-15[X^{2}]^{3}-120[X][X^{2}][X^{3}]\\\nonumber
&&\quad\quad\quad\quad\quad\quad\quad-90[X]^{2}[X^{4}]+40[X^{3}]^{2}+90[X^{2}][X^{4}]+144[X][X^{5}]-120[X]^{6}).
\end{eqnarray}
with $[X^{k}]\equiv\mbox{tr}(X^{k})$, for conciseness. Clearly, the above comprise the main algebraic invariants up to sixth order.

\end{document}